\def\isarxiv{1} 

\ifdefined\isarxiv
\documentclass[11pt]{article}

\usepackage[numbers]{natbib}

\else

\documentclass{article}

\PassOptionsToPackage{numbers}{natbib}

\usepackage[final]{neurips_2024}

\usepackage[utf8]{inputenc} 
\usepackage[T1]{fontenc}    
\usepackage{hyperref}       
\usepackage{url}            
\usepackage{booktabs}       
\usepackage{amsfonts}       
\usepackage{nicefrac}       
\usepackage{microtype}      
\usepackage{xcolor}         

\fi

\usepackage{amsmath}
\usepackage{amsthm}
\usepackage{amssymb}
\usepackage{algorithm}
\usepackage{subfig}
\usepackage{algpseudocode}
\usepackage{graphicx}
\usepackage{grffile}
\usepackage{wrapfig,epsfig}
\usepackage{url}
\usepackage{xcolor}
\usepackage{epstopdf}

\allowdisplaybreaks

\ifdefined\isarxiv
\renewcommand*{\citet}{\cite}
\renewcommand*{\citep}{\cite}

\usepackage{tikz}
\usepackage{hyperref}  
\hypersetup{colorlinks=true,citecolor=blue,linkcolor=blue} 
\usetikzlibrary{arrows}
\usepackage[margin=1in]{geometry}

\else
\renewcommand*{\citet}{\cite}
\renewcommand*{\citep}{\cite}
\usepackage{microtype}
\usepackage{hyperref}
\definecolor{mydarkblue}{rgb}{0,0.08,0.45}
\hypersetup{colorlinks=true, citecolor=mydarkblue,linkcolor=mydarkblue}

\fi

\theoremstyle{plain}
\newtheorem{theorem}{Theorem}[section]
\newtheorem{lemma}[theorem]{Lemma}
\newtheorem{definition}[theorem]{Definition}

\newtheorem{fact}[theorem]{Fact}
\newtheorem{remark}[theorem]{Remark}

\newcommand{\wt}{\widetilde}

\newcommand{\R}{\mathbb{R}}

\newcommand{\Tmat}{{\cal T}_{\mathrm{mat}}}

\makeatletter
\newcommand*{\RN}[1]{\expandafter\@slowromancap\romannumeral #1@}
\makeatother

\usepackage{lineno}

\begin{document}

\ifdefined\isarxiv

\date{}

\title{A Tighter Complexity Analysis of SparseGPT}
\author{
Xiaoyu Li\thanks{\texttt{
xli216@stevens.edu}. Stevens Institute of Technology.}
\and
Yingyu Liang\thanks{\texttt{
yingyul@hku.hk}. The University of Hong Kong. \texttt{
yliang@cs.wisc.edu}. University of Wisconsin-Madison.} 
\and
Zhenmei Shi\thanks{\texttt{
zhmeishi@cs.wisc.edu}. University of Wisconsin-Madison.}
\and 
Zhao Song\thanks{\texttt{ magic.linuxkde@gmail.com}. The Simons Institute for the Theory of Computing at the University of California, Berkeley.}
}

\else

\title{A Tighter Complexity Analysis of SparseGPT} 

\author{%
  Xiaoyu Li \\
  Stevens Institute of Technology \\
  \texttt{xli216@stevens.edu} \\
    \And    
  Yingyu Liang \\
  The University of Hong Kong \\ 
  University of Wisconsin-Madison  \\
  \texttt{
yingyul@hku.hk, yliang@cs.wisc.edu} \\
  \And
  Zhenmei Shi \\
  University of Wisconsin-Madison \\
  \texttt{zhmeishi@cs.wisc.edu} \\
  \And
  Zhao Song \\
  The Simons Institute for the Theory of Computing \\ at the University of California, Berkeley \\
  \texttt{magic.linuxkde@gmail.com} \\
  \And
}

\maketitle 
\fi

\ifdefined\isarxiv
\begin{titlepage}
  \maketitle
  \begin{abstract}

In this work, we improved the analysis of the running time of SparseGPT [Frantar, Alistarh ICML 2023] from $O(d^{3})$ to $O(d^{\omega} + d^{2+a+o(1)} + d^{1+\omega(1,1,a)-a})$ for any $a \in [0, 1]$, where $\omega$ is the exponent of matrix multiplication. In particular, for the current $\omega \approx 2.371$ [Alman, Duan, Williams, Xu, Xu, Zhou 2024], our running time boils down to $O(d^{2.53})$. This running time is due to the analysis of the lazy update behavior in iterative maintenance problems such as [Deng, Song, Weinstein 2022; Brand, Song, Zhou ICML 2024].
  \end{abstract}
  \thispagestyle{empty}
\end{titlepage}

\newpage

\else

\begin{abstract}

In this work, we improved the analysis of the running time of SparseGPT [Frantar, Alistarh ICML 2023] from $O(d^{3})$ to $O(d^{\omega} + d^{2+a+o(1)} + d^{1+\omega(1,1,a)-a})$ for any $a \in [0, 1]$, where $\omega$ is the exponent of matrix multiplication. In particular, for the current $\omega \approx 2.371$ [Alman, Duan, Williams, Xu, Xu, Zhou 2024], our running time boils down to $O(d^{2.53})$. This running time is due to the analysis of the lazy update behavior in iterative maintenance problems such as [Deng, Song, Weinstein 2022; Brand, Song, Zhou ICML 2024].
\end{abstract}

\fi

\section{Introduction}\label{sec:intro} 

Large Language Models (LLMs) have been widely applied in many AI applications. 
Accelerating its inference speed is crucial to reduce the latency time during user usage. 
One of the recent brilliant works, SparseGPT~\citep{fa23}, uses calibration data to prune the parameters of GPT-family models~\citep{bmr+20, zrg+22}, by using the optimal brain damage technique~\citep{lds89, hs92}. 
Their algorithm (Algorithm~\ref{alg:sparse_gpt}) can prune at least 50\% parameters with structure patterns, while the perplexity increase is negligible. Thus, SparseGPT can reduce the running time and GPU memory usage while keeping high performance for LLMs' applications.  

SparseGPT~\citep{fa23} claims their main pruning algorithm (Algorithm~\ref{alg:sparse_gpt}) takes running time $O(d^3)$ where $d$ is the model's hidden feature dimensions. In our work, we improved the running time analysis to get a tighter running time complexity bound, i.e., $O(d^{2.53})$, for Algorithm~\ref{alg:sparse_gpt}. Formally, we state our main results as follows.

\begin{theorem}[Main result (Restatement of Theorem~\ref{thm:main})]\label{thm:main_informal}
    Let lazy update block size $B = d^a$ for any $a \in [0, 1]$. Then Procedure \textsc{SparseGPT} in Algorithm~\ref{alg:sparse_gpt} achieves the running time
    \begin{align*}
        O(d^{\omega} + d^{2+a+o(1)} + d^{1+\omega(1,1,a)-a}).
    \end{align*}
\end{theorem}
For the current $\omega \approx 2.371$, according to the Table~1 in~\citet{adw+24}, we should choose $a\approx 0.5275$ in order to balance the terms $d^{2+a+o(1)}$ and $d^{1+\omega(1,1,a)-a}$. Then the running time boils down to $O(d^{2.53})$, which is better than $O(d^3)$ as claimed in~\citet{fa23}.

The key technique in our analysis is to compute the complexity of SparseGPT by using lazy update. 
It comes from an interesting fact of fast rectangular matrix multiplication: the time complexity of multiplying a $d \times d$ matrix by a
$d \times 1$ matrix is the same as the times complexity of multiplying a $d \times d$ matrix by a $d \times d^{a}$ matrix for any nonnegative $a \leq \alpha$ where $\alpha$ is the dual exponent of matrix multiplication and currently $\alpha \approx 0.321$~\citep{wxxz24, lg24, lgu18}. 
The lazy update has been widely used to speed up the training of neural networks~\citep{clp+20, bpsw21, szz24} and maintain dynamic attention in large language models~\citep{bsz24}.
The high-level intuition of lazy update is that during iterative one-rank updates for a $d \times d$ matrix, 
we use an external matrix $d\times d^a$ to save the input for updates while 
we conduct fast matrix multiplication only when $d^a$ updates happen. Based on the lazy update idea, we achieve our tighter complexity bound.  
\section{Related Work}\label{sec:related}

\subsection{Model Acceleration}
Model acceleration is critical and urgently needed in many practical applications. There are many techniques for model acceleration. One line of work is to change model architecture to support fast inference, e.g., Mamba~\citep{gd23,dg24}, PolySketchFormer~\citep{kmz23}, Linearizing Transformers~\citep{zhdk23,zbkr24,mvk+24}, Hopfield Model~\citep{hu2024nonparametric,hu2024outlier,wu2024uniform,xu2024bishop,hu2024computational,wu2023stanhop,hu2023sparse,hwl24} and so on. 
Another line of work is accelerating model computation on system level, e.g., FlashAttetnion~\citep{dfe+22,d23,sbz+24,sy24,lls+24_io}, parallel decoding~\citep{ssu18}, quantization~\citep{zlc+24,tzz+24,ltt+24, zdh24} and so on. 
To accelerate LLMs training and inference, there is a line of work to approximate attention matrix computation~\citep{as23,as24_arxiv,hjk+23,zhmk24,lssz24a,pmn+23,ctwc24,lssy24,lls+24b, gswy23,lssy24,lss+24,cll+24,hwsl24,cls+24} in almost linear time. 
Some specific technique is developed to accelerate LLMs generation, e.g., KV-Cache compression~\citep{gzl+23,ldlg23,xtc+23,zsz+24,lwd+23,dyz+24,scy+24,smn+24} and speculative decodings~\citep{clg+24,lch+24, scy+24, esl+24, dms24}.  

\subsection{Model Pruning}
Pruning is a technique aimed at reducing the number of weights in a neural network by selectively removing certain neurons~\citep{lds89, hs92}. This approach has gained considerable attention in recent years as a method to enhance the efficiency and scalability of deep learning models~\citep{hptd15, fc18, lat19, wzg19, bgfg20, bmbe20, lz20, tkyg20, cjd+21, habn+21, hci+21, jcr+22, fa22, fa23, slbk24}.
Pruning can be categorized based on its stage in the training process: pre-training pruning and post-training pruning. Pre-training pruning involves pruning the network at initialization. \citet{fc18} demonstrated that a neural network pruned at initialization could be trained to achieve performance comparable to a dense model, a phenomenon referred to as the lottery ticket hypothesis. This discovery motivated a line of research focused on developing methods to reduce the computational cost of pruning neural networks at initialization~\citep{lat19, wzg19, lz20, tkyg20, cjd+21}. More recently, \citet{ylg+23} provided theoretical evidence that pre-training pruning can enhance a model's generalization ability.
Post-training pruning, initially popularized through its application in quantization~\citet{navb+20, hnh+21, lgt+21, slbk24}, was later extended to pruning by \citet{hci+21, fa22, kks+22, lls+24_prune}. Post-training pruning aims to compress a well-optimized model using a small set of calibration data. This process typically involves layer-wise pruning of the neural network. Notably, \citet{hci+21} proposed a provable and efficient method for compressing a model by ``stitching together'' individually compressed layers.

\subsection{Lazy Update}

In recent years, the lazy update idea used in iterative maintenance problems has emerged as a crucial technique to effectively and efficiently solve various optimization problems, including linear programming~\citep{cls19, bra20, blss20, jswz21, sy21, lsz+23}, semi-definite programming~\citep{jkl+20, hjs+22, gs22, syz23}, empirical risk minimization~\citep{lsz19, gsz23, qszz23}, cutting plane methods~\citep{jlsw20}, neural network training~\citep{clp+20, bpsw21, szz24}, discrepancy minimization~\citep{dsw22}, dynamic attention problems~\citep{bsz24} and so on. 

We highlight several previous papers that share similar running time complexity. For comparison, we assume the input size of these problems is dominated by a parameter $d$. In the line of developments of fast linear program solvers, \citet{cls19} first introduced the lazy update idea to develop efficient dynamic inverse structures for implementing interior point methods in solving linear programming problems, and it achieves the running time $O^*(d^\omega + d^{2.5-a/2} + d^{1.5+a})$. \citet{sy21} uses a different method to achieve linear programming in the same complexity $O^*(d^\omega + d^{2.5-a/2} + d^{1.5+a})$ by incorporating the sparse sketching matrix to speed up the online matrix-vector multiplication. \citet{jswz21} improves the linear programming to $O^*(d^\omega + d^{2.5 - a/2} + d^{1.5 + a - \wt{a}/2} + d^{0.5+a+(\omega-1)\wt{a}})$ time where $\wt{a} \in [0, \alpha a]$ by designing a two-level lazy update framework for efficiently maintaining a projection matrix.
Later, \citet{lsz19} generalizes the algorithm for linear programming to empirical risk minimization, which also takes $O^*(d^\omega + d^{2.5-a/2} + d^{1.5+a})$ time. Different from~\citet{cls19}, which employs a stochastic central path method that updates weights using a random sparse vector, \citet{lsz19} introduces a robust deterministic central path method. Additionally, \citet{lsz19} proposes an efficient data structure capable of handling updates efficiently even when the weight update vector is dense. Further, \citet{qszz23} uses a different method to achieve empirical risk minimization in the same complexity $O^*(d^\omega + d^{2.5-a/2} + d^{1.5+a})$ via online projection matrix-vector multiplication. 

There are numerous other applications of the lazy update technique. For instance, \citet{dsw22} employs this approach to solve a subroutine of the discrepancy minimization problem in $O^*(d^\omega + d^{2+a} + d^{1 + \omega(1,1,a) - a})$. This is achieved by designing a data structure that efficiently implements the iterative Edge-Walk partial-coloring algorithm proposed by \citet{lm15}, utilizing the lazy update concept. More recently, \citet{bsz24} demonstrates how the lazy update concept can be applied to a dynamic attention problem, achieving an amortized update time of $O(d^{\omega(1,1,a)-a})$.
\section{Preliminary}\label{sec:preli}
\subsection{Notations.}
For a matrix $A$, we denote its transpose by $A^\top$. We use $I_{d\times d}$ to denote a $d\times d$ identity matrix. We use $\mathbf{1}_{n \times d}$ to denote an $n \times d$ matrix where all entries are ones and $\mathbf{0}_{n \times d}$ to denote an $n \times d$ matrix where all entries are zeros. For a matrix $A$, the notation $A_{[i_1, i_2], [j_1, j_2]}$ refers to the submatrix of $A$ corresponding to the rows from $i_1$ to $i_2$ (including $i_1$ and $i_2$) and columns from $j_1$ to $j_2$ (including $j_1$ and $j_2$). When we write $A_{*, [j_1, j_2]}$, it denotes the submatrix of $A$ that includes all rows and restricts the columns to those from $j_1$ to $j_2$. For two matrices $A, B$ of the same size, we denote the Hadamard product of $A$ and $B$ by $A \circ B$. We use
$O^*(f(d))$ to hide $f(d)^{o(1)}$ factor.

\subsection{Definitions and Facts} 
We now introduce some key definitions and key facts.
\begin{definition}
    For three positive integers $d_1, d_2, d_3$, we use $\Tmat(d_1, d_2, d_3)$ to denote the time of multiplying a $d_1 \times d_2$ matrix with a $d_2 \times d_3$ matrix. 
\end{definition}

The following fact shows that the order of $d_1, d_2, d_3$ only results in a constant factor difference. 

\begin{fact}[\citet{bcs13,bla13}]
    It holds that
    \begin{align*}
        \Tmat(d_1, d_2, d_3) = O(\Tmat(d_1, d_3, d_2)) = O(\Tmat(d_2,d_1,d_3)).
    \end{align*}
\end{fact}
We provide a definition of $\omega(\cdot,\cdot,\cdot)$ here. 
\begin{definition}
\label{def:omega_alpha}
For $a,b,c>0$, we use $d^{w(a,b,c)}$ to denote the time complexity of multiplying a $d^a \times d^b$ matrix with a $d^b \times d^c$ matrix. We define $\omega := \omega(1,1,1)$ as the exponent of matrix multiplication. We use $\alpha$ to denote the dual exponent of matrix multiplication, which is the largest value such that $\omega(1, \alpha, 1) = 2 + o(1)$. 
\end{definition}

In other words, $\omega$ means that multiplying two $d \times d$ matrices require time $O(d^\omega)$, and $\alpha$ is the largest number such that we can multiply a $d \times d^\alpha$ matrix with a $d^\alpha \times d$ in the subquadratic time.

\begin{lemma}[\citet{adw+24, wxxz24, lg24}]
    Currently, we have $\omega \approx 2.371$ and $\alpha \approx 0.321$. 
\end{lemma}

\begin{algorithm}[!ht]\caption{The SparseGPT algorithm (Algorithm~1 in~\citet{fa23}).}
\label{alg:sparse_gpt}
\begin{algorithmic}[1]
\Procedure{SparseGPT}{$p\in[0,1]$, $W \in \R^{d \times d}$, $X \in \R^{d \times d}$, $B \in \mathbb N_+$, $B_s \in \mathbb N_+$, $\lambda > 0$}
\State \Comment{Pruning ration $p \in [0,1]$}
\State \Comment{Weight matrix $W \in \R^{d \times d}$}
\State \Comment{Input feature matrix $X \in \R^{d \times d}$}
\State \Comment{Lazy update block size $B \in \mathbb N_+$, $B=d^a$ for any $a\in[0,1]$}
\State \Comment{Adaptive mask size $B_s \in  \mathbb N_+$}
\State \Comment{Regularization parameter $\lambda > 0$}
\State $M, ~ E \gets {\bf 1}_{d \times d}, ~ {\bf 0}_{d \times B}$ \label{line:init_m} \Comment{$O(d^2)$}
\State $\wt{H} \gets (XX^\top + \lambda I_{d \times d})^{-1}$  \Comment{$O(d^\omega)$ by Lemma~\ref{lem:run_time_h}} \label{line:init_h} 
\For{$i = 0, B, 2B, \ldots, \lfloor \frac{d}{B} \rfloor B$}
\For{$j = i + 1, \ldots, i+B$}
\If{$j \bmod B_s = 0$} \label{line:check_j} \Comment{$O(d)$}
\State $M_{*,[j, j+B_s]} \gets \textsc{MaskSelect}(p, W_{*, [j, j+B_s]}, \wt{H}, j - 1)$ \label{line:update_m}
\Comment{$O(d^2 \log d)$ by Lemma~\ref{lem:run_time_mask}}
\EndIf
\State $E_{*,j-i} \gets (\mathbf{1}_{d \times 1}-M_{*,j})\circ W_{*,j}/\wt{H}_{j, j}$ \label{line:update_e_2} \Comment{$O(d^2)$}
\State $W_{*,[j, i+B]} \gets W_{*, [j, i+B]} - E_{*,j-i} \wt{H}_{j,[j, i+B]}$ \label{line:update_w_inner}
\Comment{$O(d^{2+a})$ by Lemma~\ref{lem:run_time_2}}
\EndFor
\State $W_{*,[i+B, d]} \gets W_{*, [i+B,d]} - E \wt{H}_{[i, i+B],[i+B, d]}$ \label{line:update_w_outter}
\Comment{$O(d^{1+\omega(1,1,a) - a})$ by Lemma~\ref{lem:run_time_1}}
\EndFor
\State $W \gets W \circ M$ \Comment{$O(d^2)$} \label{line:final_w}
\EndProcedure
\State
\Procedure{MaskSelect}{$p \in [0,1]$, $W' \in \R^{d \times r}$, $\wt{H} \in \R^{d \times d}, s \in \mathbb{N}_+$} 
\State \Comment{Sub-weight matrix $W' \in \R^{d \times r}$; Inverse of Hessian matrix $\wt{H} \in \R^{d \times d}$}
\State \Comment{Index $s \in \mathbb{N}_+$, recording the position of $W'$ in $W$}
\State $M' \gets \mathbf{0}_{d \times r}$ \label{line:sub_init_m} \Comment{$O(dr)$}
\For{$k = 1, \ldots, r$}
\State $w \gets W'_{*, k} $ \label{line:sub_init_w} \Comment{$w\in\R^d$, $O(dr)$}
\State $w \gets (w \circ w)/ (\wt{H}_{s+k, s+k})^2$ \label{line:sub_update_w} \Comment{$O(dr)$}
\State $J \gets $ \label{line:sub_update_i} indices of top $(1-p)d$ largest entries of $w$ \Comment{$O(r \cdot d\log d)$ by sorting}
\For{$j \in J$}
\State $M'_{k, j} \gets 1$ \label{line:sub_update_m_ki} \Comment{$O(dr)$}
\EndFor
\EndFor
\State \Return{$M'$}
\EndProcedure
\end{algorithmic}
\end{algorithm}

\section{Main Results}\label{sec:analysis}

In this section, we present our principal findings. Our analysis demonstrates that SparseGPT (refer to Algorithm~\ref{alg:sparse_gpt}) attains the desired running time of $ O(d^{\omega} + d^{2+a} + d^{1+\omega(1,1,a)-a}) $ (Theorem~\ref{thm:main}). For the sake of clarity, we assume that both the weight matrix $ W $ and the input feature matrix $ X $ are of dimensions $ \mathbb{R}^{d \times d} $. However, our analysis remains valid for more general cases where $ W \in \mathbb{R}^{n \times d} $ and $ X \in \mathbb{R}^{d \times N} $, provided that $ n = O(d) $ and $N = O(d) $.

\begin{theorem}[Main result]\label{thm:main}
    Let lazy update block size $B = d^a$ for any $a \in [0, 1]$. Then Procedure \textsc{SparseGPT} in Algorithm~\ref{alg:sparse_gpt} achieves the running time
    \begin{align*}
        O^*(d^{\omega} + d^{2+a} + d^{1+\omega(1,1,a)-a}).
    \end{align*}
\end{theorem}
\begin{proof}
    We split the analysis of running time in the following
    \begin{itemize}
        \item Line~\ref{line:init_m} takes time $O(d^2)$ to initiate $M$ and takes time $O(d^{1+a})$ to initiate $E$.
        \item By Lemma~\ref{lem:run_time_h}, Line~\ref{line:init_h} takes time $O(d^\omega)$.
        \item In each iteration, Line~\ref{line:check_j} takes time $O(1)$ to check if $j \bmod B_s = 0$. Since there are $d$ iterations, the total time is $O(d)$.
        \item By Lemma~\ref{lem:run_time_mask}, Line~\ref{line:update_m} takes time $O(d^2\log d) = O(d^{2+o(1)})$ over all iterations.
        \item In each iteration, Line~\ref{line:update_e_2} takes time $O(d)$ to compute $W_{*,j}/[\wt{H}]_{jj}^2$ and takes time $O(d)$ to compute $1_{d \times 1}-M_{*,j}$ and then takes time $O(d)$ to compute the Hadamard product of them. Since there are $d$ iterations, the total time is $O(d^2)$.
        \item By Lemma~\ref{lem:run_time_2}, Line~\ref{line:update_w_inner} takes time $O(d^{2+a})$ over all iterations.
        \item By Lemma~\ref{lem:run_time_1}, Line~\ref{line:update_w_outter} takes time $O(d^{1+\omega(1,1,a) - a})$ over all iterations.
        \item Line~\ref{line:final_w} takes time $O(d^2)$ to compute the Hadamard product of two $d\times d$ matrices.
    \end{itemize}
    Summing up,  Algorithm~\ref{alg:sparse_gpt} runs in time
    \begin{align*}
        O(\underbrace{d^2}_{\mathrm{Line}~\ref{line:init_m}} + \underbrace{d^\omega}_{\mathrm{Line}~\ref{line:init_h}} + \underbrace{d}_{\mathrm{Line}~\ref{line:check_j}} + \underbrace{d^{2+o(1)}}_{\mathrm{Line}~\ref{line:update_m}} + \underbrace{d^2}_{\mathrm{Line}~\ref{line:update_e_2}} + \underbrace{d^{2+a}}_{\mathrm{Line}~\ref{line:update_w_inner}} + \underbrace{d^{1+\omega(1,1,a)-a}}_{\mathrm{Line}~\ref{line:update_w_outter}} + \underbrace{d^2}_{\mathrm{Line}~\ref{line:final_w}}).
    \end{align*}
    Absorbing the same terms and lower order terms, we get our target running time
   \begin{align*}
        O(d^{\omega} + d^{2+a+o(1)} + d^{1+\omega(1,1,a)-a}).
    \end{align*}
    The proof is complete.
\end{proof}

\begin{remark}
    The $O(d^2)$ space complexity is sufficient to implement fast matrix multiplication, which is equivalent to that of standard matrix multiplication. In fact, all algorithms derived from Strassen's original approach exhibit a $\Theta(d^2)$ space complexity. For further details, we refer readers to \cite{y84, a91}.
\end{remark}

\section{Detailed Complexity Analysis}

From now on, we let the lazy update block size $ B = d^a $ for any $ a \in [0, 1]$ without defining it in the statement of lemmas. We first analyze the running time of the subroutine Procedure \textsc{MaskSelect} in Algorithm~\ref{alg:sparse_gpt}.

\begin{lemma}\label{lem:run_time_mask}
The running time of Procedure \textsc{MaskSelect} in Algorithm~\ref{alg:sparse_gpt} is $O(rd\log d)$. Hence the running time of Line~\ref{line:update_m} in Algorithm~\ref{alg:sparse_gpt} over all iterations is $O(d^2 \log d)$.
\end{lemma}
\begin{proof}
    We first analyze the running time of Procedure \textsc{MaskSelect}.
    We split the analysis of running time in the following
    \begin{itemize}
        \item Line~\ref{line:sub_init_m} takes time $O(dr)$ to initialize $M'$.
        \item In every iteration, Line~\ref{line:sub_init_w} takes time $O(d)$ to initialize $w$. Since there are $r$ iterations, the total running time is $O(dr)$.
        \item In every iteration, Line~\ref{line:sub_update_w} takes time $O(1)$ to compute $(\wt{H}_{s+k})^2$ and then takes time $O(d)$ to compute $(w \circ w) / a$. Since there are $r$ iterations, the total running time is $O(dr)$.
        \item In every iteration, Line~\ref{line:sub_update_i} takes time $O(d\log d)$ to sort $w$ (without overwriting) and takes time $O(d)$ to read the indices of top $(1-p)d$ largest entries. Since there are $r$ iterations, the total running time is $O(r(d \log d + d)) = O(r d \log d)$.
        \item In every iteration over $k$, Line~\ref{line:sub_update_m_ki} takes time $O(d)$ to update $M'$. Since there are $r$ iterations, the total running time is $O(dr)$.
    \end{itemize}
    Hence the Procedure \textsc{MaskSelect} takes time
    \begin{align*}
        O(dr + dr + r + dr + rd\log d + dr) = O(rd\log d).
    \end{align*}

    Now, we analyze the running of Line~\ref{line:update_m} over all iterations in Algorithm~\ref{alg:sparse_gpt}. Let $r = B_s$.
    The total number calls to \textsc{MaskSelect} is $d/B_s$. Hence the running over all iterations is
    \begin{align*}
        O((d/B_s) \cdot (B_s d \log d)) = O(d^2 \log d).
    \end{align*}
    Thus we complete the proof.
\end{proof}

Next, we analyze the running time of several key steps of Procedure \textsc{SparseGPT} in Algorithm~\ref{alg:sparse_gpt}. We show that the inverse Hessian can be computed in time $O(d^\omega)$.

\begin{lemma}\label{lem:run_time_h}
    The running time of Line~\ref{line:init_h} in Algorithm~\ref{alg:sparse_gpt} is $O(d^\omega)$.
\end{lemma}
\begin{proof}
    Line~\ref{line:init_h} takes time $O(d^\omega)$ to compute $XX^\top$, takes time $O(d)$ to add $XX^\top$ and $\lambda I_{d\times d}$. Then computing the inverse of a $d \times d$ matrix takes time $O(d^\omega)$. Hence Line~\ref{line:init_h} takes time
    \begin{align*}
        O(d^\omega + d + d^\omega) = O(d^\omega),
    \end{align*}
    where it follows from $\omega \geq 2$.
\end{proof}

Now, we compute the running time of updates of $W$ in the inner iterations.

\begin{lemma}\label{lem:run_time_2}
The running time of Line~\ref{line:update_w_inner} in Algorithm~\ref{alg:sparse_gpt} over all iterations is $O(d^{2+a})$.
\end{lemma}
\begin{proof}
In every iteration, Line~\ref{line:update_w_inner} involves one matrix addition and one matrix multiplication. The size of $E_{*, j-i}$ is $d \times 1$ and the size of $\wt{H}_{j, [j, i+B]}$ is $1 \times B$. Hence multiplying $E_{*, j-i}$ with $\wt{H}_{j, [j, i+B]}$ takes time $O(dB)$. Subtracting $E_{*, j-i}\wt{H}_{j, [j, i+B]}$ from $W_{*, [j, i+B]}$ takes time $O(dB)$. Since there are $d$ iterations, the running time of Line~\ref{line:update_w_inner} in Algorithm~\ref{alg:sparse_gpt} over all iterations is 
\begin{align*}
    O(d \cdot (dB + dB)) = O(d^2B)= O(d^{2+a}),
\end{align*}
where the first step uses basic algebra and the second step is due to $B = d^a$. 
\end{proof}

Finally, we provide the analysis of the running time for updating $W$ in the outer iterations.

\begin{lemma}\label{lem:run_time_1}
The running time of Line~\ref{line:update_w_outter} in Algorithm~\ref{alg:sparse_gpt} over all iterations is $O(d^{1+\omega(1,1,a) - a})$.

\end{lemma}
\begin{proof}
In every iteration, Line~\ref{line:update_w_outter} involves one matrix addition and one matrix multiplication. The size of $E$ is $d \times B$. The size of $\wt{H}_{[i,i+B], [i+B,d]}$ is $B \times (d-(i+B))$. Multiplying $E$ with $\wt{H}_{[i,i+B], [i+B, d]}$ takes time $\Tmat(d, B, d-(i+B)) = \Tmat(d, B, d)$. Subtracting $E\wt{H}_{[i, i+B], [i+B,d]}$ from $W_{*, [i+B,d]}$ takes time $O(dB)$.
Hence, the running time of Line~\ref{line:update_w_outter} over all iterations is
\begin{align*}
\sum_{i=1}^{d/B} O(dB +  \Tmat(d, B, d))
= & ~ \sum_{i=1}^{d/B} O(\Tmat(d, B, d)) \\
= & ~ (d/B) \cdot O(\Tmat(d, B, d)) \\
= & ~ d^{1-a} \cdot O(\Tmat(d,d^a, d)) \\
= & ~ d^{1-a} \cdot O({d^{\omega(1,1,a)}}) \\
= & ~ O(d^{1 + \omega(1,1,a) - a}),
\end{align*}
where the first step is because $\Tmat(d, B, d) \geq \Omega(dB)$, the second step follows from basic algebra, the third step uses $B=d^a$, the fourth step is due to Definition~\ref{def:omega_alpha}, and the last step uses basic algebra.
\end{proof}

\section{Conclusion}\label{sec:conclusion}
We improved the complexity analysis of SparseGPT from $O(d^3)$ to $O(d^{2.53})$, using techniques from fast matrix multiplication and lazy update ideas used in iterative maintenance problems. This tighter bound demonstrates that large language models can be pruned more efficiently than previously thought. Future work could explore further improvements or extensions of these methodologies to other model compression algorithms.

\ifdefined\isarxiv
\section*{Acknowledgement}
Research is partially supported by the National Science Foundation (NSF) Grants 2023239-DMS, CCF-2046710, and Air Force Grant FA9550-18-1-0166.
\else

\fi

\ifdefined\isarxiv
\bibliographystyle{alpha}
\bibliography{ref}
\else
\bibliography{ref}
\bibliographystyle{alpha}
\fi





\end{document}